\documentclass[titlepage,11pt]{article}
\usepackage{graphicx}
\usepackage{amsmath}
\usepackage{amssymb}
\usepackage{mathptmx}
\usepackage{pifont}

\newcommand{\OMIT}[1]{} %

\newcommand{\p}{{\rm P}}
\newcommand{\np}{{\rm NP}}
\newcommand{\conp}{{\rm coNP}}

\newcommand{\lp}{{\cal LP}^*}
\newcommand{\calL}{{\cal L}}
\newcommand{\calA}{{\cal A}}
\newcommand{\I}{{\cal I}}
\newcommand{\calS}{{\cal S}}
\newcommand{\calCS}{{\cal CS}}
\newcommand{\calG}{{\cal G}}
\newcommand{\calC}{{\cal C}}
\newcommand{\eps}{{\varepsilon}}
\newcommand{\CoS}{\mathit{CoS}}
\newcommand{\poly}{\mathrm{poly}}
\newcommand{\core}{\mathrm{core}}
\newcommand{\cscore}{\textit{CS-core}}
\newcommand{\vecp}{{\mathbf p}}
\newcommand{\vecDelta}{\vec{\Delta}}
\newcommand{\w}{{\mathbf w}}
\newcommand{\CS}{\mathit{CS}}

\newtheorem{theorem}{Theorem}[section]

\newtheorem{lemma}[theorem]{Lemma}
\newtheorem{proposition}[theorem]{Proposition}
\newtheorem{example}[theorem]{Example}
\newtheorem{definition}[theorem]{Definition}

\newcommand\qedblob{\mbox{\ding{113}}}

\newenvironment{proof}{\noindent{\bf Proof.}\hspace*{1em}}{\literalqed\bigskip}
\def\literalqed{{\ \nolinebreak\hfill\mbox{\qedblob\quad}}}

\newcommand{\sproofof}[1]{\noindent{\bf Proof of {#1}.}\hspace*{1em}}
\newcommand{\eproofof}[1]{\noindent{\hspace*{0.1in} \hfil \hfill \mbox{\literalqed{} {#1}}}\quad\bigskip}

\begin{document}

\title{The Cost of Stability
in Coalitional Games\thanks{This paper's results will be presented at the
\emph{2nd International Symposium on Algorithmic Game Theory}.
A preliminary version of this paper was published in the proceedings
of the \emph{8th International Joint Conference on Autonomous Agents and
Multiagent Systems}~\cite{aamas-cos}.}}

\author{Yoram Bachrach\thanks{Microsoft Research, Cambridge, United Kingdom.}
\and
Edith Elkind\thanks{University of Southampton, United Kingdom.}$^{,}$
\thanks{Nanyang Technological University, Singapore.}
\and
Reshef Meir\thanks{School of Engineering and Computer Science,
Hebrew University, Jerusalem, Israel.} 
\and
Dmitrii Pasechnik$^4$
\and
Michael Zuckerman$^5$
\and
J\"org Rothe\thanks{Institut f{\"{u}}r Informatik, 
Heinrich-Heine-Universit{\"a}t D{\"u}sseldorf,
Germany.} 
\and Jeffrey S. Rosenschein$^5$
}

\date{July 24, 2009}

\maketitle
\begin{abstract}
{\begin{small}
  A key question in cooperative game theory is that of coalitional
  stability, usually captured by the notion of the \emph{core}---the set of outcomes
  such that no subgroup of players has an incentive 
  to deviate. However,
  some coalitional games have empty cores, and any outcome in such
  a game is unstable.

  In this paper, 
  we investigate the possibility of stabilizing a coalitional game 
  by using external payments. We consider a scenario where an external party, which
  is interested in having the players work
  together, offers a supplemental payment to the grand coalition (or, more generally, 
  a particular coalition structure). This payment is conditional 
  on players not deviating from their coalition(s). The sum of this payment 
  plus the actual gains of the coalition(s) may then be divided among the
  agents so as to promote stability. We define the \emph{cost of
  stability (CoS)} as the minimal external payment that stabilizes the game.

  We provide general bounds on the cost of stability in several
  classes of games, and explore its algorithmic properties. To develop
  a better intuition for the concepts we introduce, we provide
  a detailed algorithmic study of the cost of stability in weighted
  voting games, a simple but expressive class of games which can
  model decision-making in political bodies, and cooperation in
  multiagent settings. 
  Finally, we extend our model and results to games with coalition structures.
\end{small}
}
\end{abstract}

\section{Introduction}
\label{l_introduction}

In recent years, 
algorithmic game theory, an emerging field that combines
computer science, game theory and social choice, has received
much attention from the multiagent 
community~\cite{Rosenschein:1985:IJCAI,Ephrati:1991:AAAI,wellman:1995,sandholm1995}. 
Indeed, 
multiagent systems research focuses on designing 
intelligent agents, i.e., entities that can
coordinate, cooperate and negotiate without requiring human
intervention. In many application domains, such agents 
are \emph{self-interested}, i.e., they are built to maximize the rewards obtained 
by their creators. Therefore, these agents can be modeled naturally using
game-theoretic tools. Moreover, as agents often have to function
in rapidly changing environments, computational
considerations are of great concern to their designers as well.

In many settings, such as online auctions and other types
of markets, agents act individually. In this case, the standard
notions of noncooperative game theory, such as \emph{Nash equilibrium}
or \emph{dominant-strategy equilibrium}, provide a prediction of the outcome of the interaction.
However, another frequently occurring type of scenario is that agents
need to form teams to achieve their individual goals. 
In such domains, the
focus turns from the interaction between single agents to the
capabilities of subsets, or \emph{coalitions}, of agents.
Thus, a more appropriate modeling toolkit for this setting
is that of \emph{cooperative}, or \emph{coalitional}, 
game theory~\cite{tijs}, 
which 
studies what coalitions are most likely to arise, and
how their members distribute the gains from cooperation.  
When agents are self-interested, the latter question is obviously
of great importance. Indeed, the
\emph{total} utility generated by the coalition is of little interest
to individual agents; rather, each agent 
aims to maximize her own utility.
Thus, a \emph{stable} coalition can be formed
only if the gains from cooperation can be distributed
in a way that satisfies all agents. 

The most prominent solution concept that aims to formalize
the idea of stability in coalitional games 
is the \emph{core}. Informally,  
an \emph{outcome} of a coalitional game is a \emph{payoff vector}
which for each agent lists her share of the profit of the 
\emph{grand coalition},
i.e., 
the coalition that includes all agents.
An outcome is said to be in the core if it distributes
gains so that no subset of 
agents has an incentive to abandon the grand coalition
and form a coalition of their own. 
It can be argued that 
the concept of the core captures the intuitive
notion of stability in cooperative settings.
However, it has an important drawback:
the core of a game may be empty. In games with
empty cores, any outcome is unstable, and therefore
there is always a group of agents that is tempted
to abandon the existing plan.
This observation 
has triggered the invention of less demanding solution
concepts, such as $\eps$-core and the least core, 
as well as an interest in noncooperative approaches
to identifying stable outcomes in coalitional 
games~\cite{chatterjee93,okada96}.

In this paper, we approach this issue from a different perspective.
Specifically, we examine the possibility of stabilizing the
outcome of a game using external payments. Under this model, an
external party (the \emph{center}), which can be seen as a central authority
interested in stable functioning
of the system, 
attempts to incentivize a coalition of agents to
cooperate in a stable manner. This party does this by offering
the members of a coalition a supplemental payment 
if they cooperate. This external
payment is given to the coalition as a whole, and is provided only if
this coalition is formed. 

Clearly, when the supplemental payment is large enough, the resulting
outcome is stable: the profit that the deviators
can make on their own is dwarfed by the subsidy they could receive
by sticking to the prescribed solution. However, normally
the external party would want to minimize its expenditure.
Thus, in this paper we define and study the \emph{cost of stability}, 
which is the minimal supplemental payment that is required to 
ensure stability in a coalitional game. We start by considering
this concept in the context where the central authority
aims to ensure that \emph{all} agents cooperate, i.e., it
offers a supplemental payment in order to stabilize the grand
coalition. We then extend our analysis to the setting where
the goal of the center is the stability of 
a \emph{coalition structure}, i.e., a partition of all
agents into disjoint coalitions. In this setting, the center
does not expect the agents to work as a single team, but 
nevertheless wants each individual team to be immune
to deviations. Finally, we consider the scenario where
the center is concerned with the stability of a particular
coalition within a coalition structure. This model is appropriate 
when the central
authority wants a particular group of agents to work together, 
but is indifferent to other agents switching coalitions.

We first provide bounds on the cost of stability in general coalitional 
games. We then show that for some interesting special cases, 
such as super-additive games, these bounds can be improved considerably.  
We also propose a general algorithmic technique
for computing the cost of stability. Then, 
to develop a better understanding of the concepts
proposed in the paper, we apply them in the context of
\emph{weighted voting games} (WVGs), a simple but powerful
class of games that have been used to model cooperation 
in settings as diverse as, on the one hand, decision-making in political bodies
such as the United Nations Security Council and 
the International Monetary Fund
and, on the other hand,  
resource allocation in multiagent systems.
For such games, we are able to obtain a complete 
characterization of the cost of stability from an
algorithmic perspective.

The paper is organized as follows. In Section~\ref{l_preliminaries},
we provide the necessary background on coalitional games.
In Section~\ref{l_def}, we
formally define the cost of stability for the setting
where the desired outcome is the grand coalition, 
prove bounds on the cost of stability, and outline a general
technique for computing it.
We then focus on the computational aspects of the cost of stability 
in the context of our selected domain, i.e., weighted voting games.
In Section~\ref{l_exact}, 
we demonstrate that computing the cost of stability in such games
is $\conp$-hard if the weights are given in binary. On the other
hand, for unary weights, we provide an efficient algorithm
for this problem. 
We also investigate whether the cost of stability can be efficiently
approximated. 
In Section~\ref{l_approx}, we answer this question positively
by describing 
a fully polynomial-time approximation scheme (FPTAS)
for our problem.
We complement this result by showing that, by distributing 
the payments in a very natural manner, we get within a factor
of 2 of the optimal adjusted gains, i.e., the sum of the value
of the grand coalition and the external payments.
While this method of allocating payoffs does not necessarily minimize
the center's expenditure, the fact that it is both easy to
implement and has a bounded worst-case performance
may make it an attractive proposition in certain settings.  
In Section~\ref{l_cs}, we extend our discussion to the setting
where the center aims to stabilize an arbitrary coalition structure, 
or a particular coalition within it, rather than the grand coalition.
We end the paper with a discussion of related work and some conclusions.

\section{Preliminaries}
\label{l_preliminaries}

Throughout this paper, given a vector $\vec{x} = (x_1, \dots, x_n)$
and a set $C\subseteq\{1, \dots, n\}$ we write $x(C)$
to denote $\sum_{i\in C}x_i$.
\begin{definition}
\label{l_def_coalitional_game}
A \emph{(transferable utility) coalitional game} $G=(I, v)$ is given by
a set of \emph{agents} (synonymously, \emph{players}) $I=\{1, \dots, n\}$ 
and a \emph{characteristic function}
$v : 2^I \rightarrow \mathbb{R}^+ \cup \{0\}$ that for any subset (coalition) 
of agents lists the total utility these agents
achieve by working together. We assume $v(\emptyset) = 0$. 
\end{definition}

A coalitional game $G=(I, v)$ is called \emph{increasing} 
if for all coalitions $C' \subseteq C$ we have $v(C') \leq v(C)$, 
and \emph{super-additive} if for all disjoint
coalitions $C, C' \subseteq I$ we have $v(C)+v(C') \leq v(C \cup C')$.
Note that since $v(C)\ge 0$ for any $C\subseteq I$, all super-additive 
games are increasing. 
A coalitional game $G=(I, v)$ is called \emph{simple} if 
it is increasing and $v(C)\in\{0, 1\}$
for all $C\subseteq I$. In a simple game, 
we say that a coalition $C \subseteq I$ \emph{wins} 
if $v(C)=1$, and \emph{loses} if $v(C)=0$. Finally, a coalitional
game is called \emph{anonymous} if $v(C)=v(C')$ for any $C, C'\subseteq I$
such that $|C|=|C'|$. A particular class of simple games considered in this paper is that
of \emph{weighted voting games} (WVGs). 

\begin{definition}
  \label{l_def_wvg} A \emph{weighted voting game} is a simple
  coalitional game given by a set of agents $I = \{1, \ldots, n\}$, a
  vector $\w = (w_1,\ldots ,w_n)$ of nonnegative weights,
  where $w_i$ is agent $i$'s weight, and a threshold~$q$. 
  The \emph{weight of a coalition} $C \subseteq I$ is 
  $w(C) = \sum_{i \in C} w_i$.
  A coalition $C$ \emph{wins the game} (i.e., $v(C) =
  1$) if $w(C) \geq q$, and \emph{loses the game} (i.e., $v(C) = 0$)
  if $w(C) < q$.
\end{definition}

We denote the WVG with the weights 
$\w = (w_1, \ldots, w_n)$ and the threshold $q$ as $[\w;q]$ or
$[w_1,\ldots,w_n;q]$. Also, we set $w_{\max}=\max_{i\in I}w_i$. 
It is easy to see that WVGs are simple
games; however, they are not necessarily super-additive.
Throughout this paper, we assume that $w(I)\ge q$, i.e., the grand coalition
wins.

The characteristic function of a coalitional game defines only the
\emph{total} gains a coalition achieves, but does not offer a way of
distributing them among the agents. Such a division is called an
imputation (or, sometimes, a payoff vector).
\begin{definition}
  \label{l_def_imputation} 
  Given a coalitional game $G=(I, v)$, a vector
  $\vecp=(p_1,\ldots ,p_n)\in \mathbb{R}^n$ is called an \emph{imputation 
  for $G$} if it satisfies $p_i \ge v( \{i \})$
  for each $i$, $1 \leq i \leq n$, and
  $\sum_{i=1}^n p_i = v(I)$. 
  We call $p_i$ the \emph{payoff of agent $i$}; the \emph{total payoff
  of a coalition $C\subseteq I$} is given by~$p(C)$.
  We write $\I(G)$ to denote the set of all imputations for $G$.
\end{definition}

For an imputation to be stable, it should be the case 
that no subset of players has an incentive
to deviate. Formally, we say that a coalition $C$ \emph{blocks} an
imputation $\vecp=(p_1,\ldots, p_n)$ if $p(C) < v(C)$.
The \emph{core} of a coalitional game $G$ is defined
as the set of imputations not blocked by any coalition, i.e.,
$\core(G) = \{ \vecp\in\I(G) \mid
p(C) \ge v(C)\ \text{for each } C \subseteq I\}$.
An imputation in the core guarantees the stability of the grand coalition.
However, the core can be empty. 

In WVGs, and, more generally, in simple games, 
one can characterize the core 
using the notion of veto agents, i.e.,
agents that are indispensable
for forming a winning coalition. Formally, given a simple coalitional
game $G=(I, v)$, an agent $i\in I$ is said to be a \emph{veto agent}
if for all coalitions $C \subseteq I\setminus\{i\}$ we have $v(C) = 0$. 
The following is a folklore result regarding nonemptiness of the core.

\begin{theorem}
\label{veto}
Let $G=(I, v)$ be a simple coalitional game. If there are no veto agents in
$G$, then the core of $G$ is empty. Otherwise, let $I'=\{i_1,\ldots, i_m\}$
be the set of veto agents in $G$. Then the core of $G$
is the set of imputations
that distribute all the gains among the veto agents only, i.e., 
\[
\core(G)=\{\vecp\in\I(G) \mid p(I') = 1\}.
\]
\end{theorem}

So far, we have tacitly assumed that the only possible outcome 
of a coalitional game is the formation of the grand coalition.
However, often it makes more sense for the agents to
form several disjoint coalitions, each of which can focus on 
its own task. For example, WVGs can be used
to model the setting where each agent has a certain amount of resources
(modeled by her weight), and there are a number of identical tasks
each of which requires a certain amount of these resources 
(modeled by the threshold)
to be completed. In this setting, the formation of the grand coalition
means that only one task will be completed, even if there are enough
resources for several tasks. 

The situation when agents can split into teams to work on several
tasks simultaneously can be modeled using the notion of a
coalition structure, i.e., a partition of the set
of agents into disjoint coalitions. Formally, we say that
$\CS=(C^1, \dots, C^m)$ is a \emph{coalition structure} 
over a set of agents $I$ if $\bigcup_{i=1}^m C^i=I$ and
$C^i\cap C^j=\emptyset$ for all $i\neq j$; we write $\CS\in\calCS(I)$.
Also, we overload notation by writing $v(\CS)$ to denote
$\sum_{C^j\in \CS}v(C^j)$.
If coalition structures are allowed, an outcome of a game
is not just an imputation, but a pair $(\CS, \vecp)$, where
$\vecp$ is an imputation for the coalition structure $\CS$, 
i.e., $\vecp$ distributes the gains of every coalition in $\CS$
among its members. Formally, we say that $\vecp=(p_1, \dots, p_n)$ 
is an \emph{imputation for a coalition structure} 
$\CS=(C^1, \dots, C^m)$ in a game $G=(I, v)$
if $p_i\ge 0$ for all $i$, $1 \leq i \leq n$, and $p(C^j)=v(C^j)$
for all $j$, $1 \leq j \leq m$; we write $\vecp\in\I(\CS, G)$.
We can also generalize the notion of the core introduced
earlier in this section to games with coalition structures.
Namely, given a game $G=(I, v)$, we say that an outcome 
$(\CS, \vecp)$ is in the \emph{CS-core of $G$} if
$\CS$ is a coalition structure over $I$, $\vecp\in\I(\CS, G)$
and $p(C)\ge v(C)$ for all $C\subseteq I$; we write
$(\CS, \vecp)\in\cscore(G)$.
Note that if $\vecp$ is in the core of $G$ then
$(I, \vecp)$ is in the CS-core of $G$; however, the converse
is not necessarily true. 

\section{The Cost of Stability}
\label{l_def}

In many games, forming the grand coalition 
maximizes social welfare; this happens, for example, in
super-additive games.
However, the core of such games may still be empty. 
In this case, it would be impossible to
distribute the gains of the grand coalition in a stable way, so it
may fall apart despite being socially optimal. 
Thus, an external party, such as a benevolent
central authority, may want to incentivize the agents to cooperate, e.g., 
by offering
the agents a supplemental payment $\Delta$ if they stay in the grand coalition. 
This situation can be modeled as an \emph{adjusted coalitional game}
derived from the original coalitional game $G$.

\begin{definition}\label{l_adjusted_game} 
Given a coalitional game $G=(I, v)$ and $\Delta \geq 0$,
the \emph{adjusted coalitional game} 
$G(\Delta)=(I, v')$ is given by
$v'(C) = v(C)$ for $C \neq I$, and
$v'(I) = v(I) + \Delta$.
\end{definition}

We call $v'(I) = v(I) + \Delta$ the \emph{adjusted
gains} of the grand coalition.
We say that a vector $\vecp\in{\mathbb R}^n$ is a \emph{super-imputation}
for a game $G=(I, v)$ if $p_i\ge 0$ for all $i\in I$ and $p(I)\ge v(I)$.
Furthermore, we say that a super-imputation $\vecp$ is \emph{stable}
if $p(C)\ge v(C)$ for all $C\subseteq I$. A super-imputation $\vecp$
with $p(I)=v(I)+\Delta$ distributes the adjusted gains, i.e., it 
is an imputation for $G(\Delta)$; it is stable if and only if 
it is in the core of $G(\Delta)$. 
We say that a  supplemental payment $\Delta$ \emph{stabilizes} the grand 
coalition in a game $G$ 
if the adjusted game $G(\Delta)$ has a nonempty core. Clearly, if $\Delta$
is large enough (e.g., $\Delta = n\max_{C\subseteq I} v(C)$), the game
$G(\Delta)$ will have a nonempty core. However, usually 
the central authority wants to spend as little money as possible.
Hence, we define the cost of stability as the \emph{smallest} external payment
that stabilizes the grand coalition.

\begin{definition}
  \label{l_def_cost_of_stability} 
  Given a coalitional game $G=(I, v)$, its \emph{cost of stability}
    $\CoS(G)$ is defined as 
\[
\CoS(G) = \inf\{\Delta\mid \Delta\ge 0\text{ }\mathrm{and}\text{ }
                 \core(G(\Delta))\neq\emptyset\}.
\]
\end{definition}

We have argued that the set 
$\{\Delta\mid \Delta\ge 0\textrm{ and }\core(G(\Delta))\neq\emptyset\}$
is nonempty. Therefore, $G(\Delta)$ is well-defined. Now, 
we
prove that this set contains 
its greatest lower bound $\CoS(G)$, i.e., that the game
$G(\CoS(G))$ has a nonempty core.
While this can be shown using a continuity argument, we will now 
give a different proof, which will also be useful
for exploring the cost of stability from an algorithmic perspective.
Fix a coalitional game $G=(I, v)$ and consider the 
following linear program $\lp$: 
\begin{eqnarray}
\label{l_lp_cos_form}
\min \Delta & &\textrm{subject to:}\nonumber\\
\Delta &\ge&  0,\\
p_i &\ge&  0 \textrm{\quad for each }i=1, \dots, n,\\
\sum_{i\in I} p_i &=& v(I)+\Delta,\\
\sum_{i\in C}p_i &\ge&  v(C)\textrm{\quad for all }C\subseteq I.
\end{eqnarray}
\vspace*{-0.50cm}

It is not hard to see that the optimal value of this linear program
is exactly $\CoS(G)$. Moreover, any optimal solution of $\lp$
corresponds to an imputation in the core of $G(\CoS(G))$ and therefore
the game $G(\CoS(G))$ has a nonempty core.

As an example, consider a uniform weighted voting game, i.e., a WVG $G=[\w; q]$
with $w_1=\cdots=w_n=w$.
We can derive an explicit formula for $\CoS(G)$.

\begin{theorem}\label{thm:identical-weights}
For a WVG $G=[w,w,\ldots,w;q]$,
we have $\CoS(G) = \frac{n}{\lceil q/w\rceil}-1$.
\end{theorem}

\begin{proof}
First, note that by scaling $w$ and $q$ we can assume that $w=1$.

Set $\Delta=\frac{n}{\lceil q\rceil}-1$ and
consider the imputation $\vecp=(p_1, \dots, p_n)$ given by
$p_i=\frac{1}{\lceil q\rceil}$ for $i$, $1 \leq i \leq n$.
Clearly, we have $p(I)=\frac{n}{\lceil q\rceil}$, so $\vecp\in G(\Delta)$.
Moreover, for any winning
coalition $C$, we have $|C|\ge \lceil q\rceil$, so
$p(C)\ge \lceil q\rceil\frac{1}{\lceil q\rceil}= 1$.
Therefore, $\vecp$ is in the core of $G(\Delta)$,
and hence $\CoS(G)\le \Delta$.

On the other hand, consider any stable super-imputation $\vecp$.
Set $s=\lceil q\rceil$. Clearly, for any coalition
$C$ with $|C|=s$ we have $p(C)\ge 1$. Now, consider
a collection of coalitions $C^1, \dots, C^n$, where
$C^i=\{i \bmod n, i+1 \bmod n, \dots, i+s-1 \bmod n\}$:
for example, we have $C^{n-1}=\{n-1, n, 1, \dots, s-2\}$.
We have $|C^i|=s$ for all $i$, $1 \leq i \leq n$,
so $p(C^1)+\dots+p(C^n)\ge n$.  Since
each player $i$ occurs in exactly $s$ of these coalitions,
we have $p(I)s=p(C^1)+\dots+p(C^n)$. Hence,
$p(I)\ge n/s=\frac{n}{\lceil q\rceil}$ and therefore
$\CoS(G)\ge \Delta$.~\end{proof}

For example, 
if $w(n-1)<q\le wn$, then $\CoS(G)=0$, i.e., 
$G$ has a nonempty core. On the other hand, 
if $w=1$, $n=3k$ and $q=2k$ for some integer $k>0$, 
i.e., $q=\frac23n$, we have $\CoS(G)=\frac32-1=\frac12$.

\subsection{Bounds on $\CoS(G)$ in General Coalitional Games}

Consider an arbitrary coalitional game $G=(I, v)$. 
Clearly, $\CoS(G)=0$ if and only if $G$ has a nonempty core.
Further, we have argued that $\CoS(G)$ is upper-bounded by $n\max_{C\subseteq I}v(C)$,
i.e., $\CoS(G)$ is finite for any fixed coalitional game. 
Moreover, the bound of $n\max_{C\subseteq I}v(C)$ is (almost) tight.
To see this, consider a (simple)
game $G'$ given by $v'(\emptyset)=0$ and $v'(C)=1$ for all $C\neq\emptyset$.
Clearly, we have $\CoS(G')=n-1$: any super-imputation that pays some
agent less than $1$ will not be stable, whereas
setting $p_i=1$ for all $i\in I$
ensures stability. Thus, the cost of stability can be quite large
relative to the value of the grand coalition.

On the other hand, we can provide a lower bound 
on $\CoS(G)$ in terms of the values of coalition structures over $I$.
Indeed, for an arbitrary coalition structure $\CS\in\calCS(I)$, 
we have $\CoS(G)\ge v(\CS)-v(I)$. To see this, note that
if the total payment to the grand coalition is less than
$(v(\CS)-v(I))+v(I)$, then for some coalition 
$C\in \CS$ it will be the case that $p(C)<v(C)$.
It would be tempting to conjecture that 
$\CoS(G)=\max_{\CS\in\calCS(I)}(v(\CS)-v(I))$. However, 
a counterexample is provided by
Theorem~\ref{thm:identical-weights}
with $w=1$, $q=\frac23n$: indeed, in this case
we have $\CoS(G)=\frac12$, yet $\max_{\CS\in\calCS(I)}(v(\CS)-v(I))=0$.
We can summarize these observations as follows.
\begin{theorem}
For any coalitional game $G=(I, v)$, we have
$$
\max_{\CS\in\calCS(I)}(v(CS)-v(I))\le \CoS(G)\le n\max_{C\subseteq I}v(C).
$$
\end{theorem}
For super-additive games, we can strengthen the upper bound considerably.
Note that in such games the grand coalition maximizes social welfare, 
so its stability is particularly
desirable. Yet, as the second part of Theorem~\ref{thm:superadd} implies, 
ensuring stability may turn out to be quite costly even in this restricted
setting.

\begin{theorem}\label{thm:superadd}
For any super-additive game $G=(I, v)$,
$|I|=n$, we have $\CoS(G)\le(\sqrt{n}-1)v(I)$,
and this bound is asymptotically tight.
\end{theorem}

\begin{proof}
Fix an arbitrary monotone super-additive game $G=(I, v)$ with 
$v(\emptyset)=0$ and $|I|=n$. 
Consider the corresponding
linear program $\lp$.
Observe that it can be re-written as 
\begin{eqnarray*}
\min \sum_{i\in I}p_i & & \textrm{subject to:}\\
p_i &\ge&  0 \textrm{\quad for }i=1, \dots, n,\\
\sum_{i\in C}p_i &\ge&  v(C) \textrm{\quad for all }C\subseteq I.
\end{eqnarray*}
The dual to this linear program 
has $2^n$ variables $\{\lambda_C\}_{C\subseteq I}$ and is given by
\begin{eqnarray*}
\max \sum_{C\subseteq I}v(C)\lambda_C & & \textrm{subject to:}\\
\lambda_C &\ge&  0 \textrm{\quad for all }C\subseteq I,\\
\sum_{C: i \in C}\lambda_C &\le&  1 \textrm{\quad for }i=1, \dots, n.
\end{eqnarray*}
That is, we have to assign ``weights'' $\lambda_C$ to all 
coalitions so that the total weight of all coalitions covering any given point is at most $1$. Our goal is to maximize
$\sum_{C\subseteq I}v(C)\lambda_C$ subject to this condition.
 
First, we claim that there exists an optimal solution
to this maximization problem that satisfies $S\cap T\neq\emptyset$
for any nonempty sets $S, T \subseteq I$ 
such that $\lambda_S>0$ and $\lambda_T>0$. 
For a contradiction, suppose that this is not the case.
Fix an arbitrary order $\prec$ on coalitions in $2^I$ such that $|S|<|T|$ implies $S\prec T$, 
and extend it to a lexicographic order on tuples of subsets of $I$ in the standard manner.
For every optimal solution $(\lambda_C)_{C\subseteq I}$ to the dual program, consider the vector 
$\xi_{(\lambda_C)_{C\subseteq I}}$ whose entries are
the subsets $C\subseteq I$ with $\lambda_C=0$, ordered according to $\prec$
(from the smallest to the largest). 
Among all optimal solutions
to the dual linear program, pick one with the lexicographically largest such vector 
and denote it by $(\lambda^*_C)_{C\subseteq I}$.   
By our assumption, there exists a pair $(S, T)$ of nonempty sets such that
$\lambda^*_S>0$ and $\lambda^*_T>0$, but $S\cap T=\emptyset$. 
Let $\eps=\min\{\lambda^*_S, \lambda^*_T\}$. Consider the vector 
$(\lambda^{**}_C)_{C\subseteq I}$ 
given by 
$$
\lambda_C^{**}=
\begin{cases}
\lambda_C^* &\text{ for }C\neq S, T, S\cup T, \\
\lambda_C^*-\eps &\text{ for }C= S, T,\\
\lambda_C^*+\eps &\text{ for }C= S\cup T.
\end{cases}
$$
First, observe that since $S$ and $T$ are disjoint,  $(\lambda^{**}_C)_{C\subseteq I}$ 
is also a feasible solution to the dual program.
Furthermore, by super-additivity we have
\begin{align*} 
\sum_{C\subseteq I}v(C)\lambda^{**}_C=
\sum_{C\subseteq I}v(C)\lambda^*_C-v(S)\eps-v(T)\eps+(v(S)+v(T))\eps\ge
\sum_{C\subseteq I}v(C)\lambda^*_C, 
\end{align*}
so $(\lambda^{**}_C)_{C\subseteq I}$ is an optimal solution to the dual program, too. 
Finally, observe that $\xi_{(\lambda^{**}_C)_{C\subseteq I}}$ is lexicographically 
greater than $\xi_{(\lambda^*_C)_{C\subseteq I}}$. Indeed, 
assume that $\eps=\lambda^*_S$ 
(a similar argument works for $\eps= \lambda^*_T$).
Then, for
$C\neq S, T, S\cup T$, 
we have $\lambda^*_C=\lambda^{**}_C$, and, moreover, 
$\lambda^*_{S}\neq 0$, $\lambda^{**}_{S}= 0$ 
and $|S\cup T|>|S|$.
This is a contradiction with our choice of $(\lambda^*_C)_{C\subseteq I}$.

Thus, there is an optimal solution $(\lambda_C)_{C\subseteq I}$ 
in which any two sets $C$ and $C'$ with 
$\lambda_C\neq 0$ and $\lambda_{C'}\neq 0$
intersect. Now, suppose that there is a set $S$ with $|S|\le \sqrt{n}$, 
$\lambda_{S}>0$. Any set $T$ with $\lambda_{T}>0$
contains one of the points in $S$. Thus, 
we have 
$$
\sum_{C\subseteq I}\lambda_Cv(C)\le v(I)\sum_{i\in S}\sum_{T: i\in T}\lambda_T
\le v(I)\sum_{i\in S}1\le \sqrt{n}v(I).
$$
On the other hand, if for any $C$ with $\lambda_C>0$ it holds that $|C|>\sqrt{n}$, 
we have 
$$
\sqrt{n}\sum_{C\subseteq I}\lambda_Cv(C)\le
\sum_{C\subseteq I}\lambda_Cv(C)|C|\le 
 v(I)\sum_{i\in I}\sum_{C: i\in C}\lambda_C\le nv(I), 
$$ 
so $\sum_{C\subseteq I}\lambda_Cv(C)\le nv(I)/\sqrt{n}=\sqrt{n}v(I)$.
Consequently, in both cases we have $\sum_{C\subseteq I}\lambda_Cv(C)\le \sqrt{n}v(I)$.
Now, since the optima of the 
dual and the original linear programs are equal, 
the optimal solution 
$(p_1, \dots, p_n)$ to the original linear program
satisfies $\sum_{i\in I}p_i\le \sqrt{n}v(I)$,
and hence $\CoS(G)\le (\sqrt{n}-1)v(I)$, as required.

To see that this bound is asymptotically
tight, consider a finite projective plane $P$ of order $q$, 
where $q$ is a prime number. It has $q^2+q+1$ points and the same number of lines, 
every line contains $q+1$ points, 
any two lines intersect, and any point belongs to exactly $q+1$ lines. Now, consider 
a simple coalitional game $G'$ whose players correspond to 
the points in $P$ and whose winning
coalitions correspond to the
sets of points in $P$ that contain a line. Observe that this game
is super-additive: since any two lines intersect, there do not exist two disjoint 
winning coalitions. Hence, for any $S, T\subseteq I$ such that $S\cap T=\emptyset$
either $v(S)=0$ or $v(T)=0$, and therefore $v(S)+v(T)\le v(S\cup T)$, as required.
On the other hand, for each line $C$, we have $\sum_{i\in C}p_i\ge 1$. Summing over all 
$q^2+q+1$ lines, 
and using the fact that each point belongs to $q+1$ lines, we obtain 
$(q+1)\sum_{i\in I}p_i\ge q^2+q+1$, i.e., $p(I)=\frac{q^2+q+1}{q+1}=q+\frac{1}{q+1}$.
Since $n=|I|=q^2+q+1$, we have $q\ge \sqrt{n}-1$, i.e., 
$\CoS(G')\ge (\sqrt{n}-2)v(I)$.~\end{proof}

For anonymous super-additive games, further improvements are possible.

\begin{theorem}\label{thm:superadd-anon}
For any anonymous super-additive game $G=(I, v)$,
we have $\CoS(G)\le 2v(I)$,
and this bound is asymptotically tight.
\end{theorem}

\begin{proof}
Fix an anonymous super-additive game $G=(I, v)$ with $|I|=n$.
Consider a super-imputation $\vecp=(p_1, \dots, p_n)$ given by
$p_i=\frac{2v(I)}{n}$. Clearly, we have $p(I)=2v(I)$. It remains
to show that $\vecp$ is in the core of the adjusted game $G(v(I))$.

For any coalition $C\subset I$, there exists an integer $k$, $1\le k\le n-1$,
such that $\frac{n}{k+1}\le |C|< \frac{n}{k}$. For this value of $k$, 
one can construct $k$ pairwise disjoint coalitions $C_1, \dots, C_k$ 
with $C_1=C$ and $|C_1|=\dots=|C_k|$.  Super-additivity then implies
that $v(C)\le \frac{v(I)}{k}$. On the other hand, we have 
$$
p(C)=|C|\frac{2v(I)}{n}\ge \frac{n}{k+1}\cdot\frac{2v(I)}{n}=\frac{2v(I)}{k+1}.
$$  
Since $\frac{2v(I)}{k+1}\ge \frac{v(I)}{k}$ for any $k\ge 1$, 
it follows that $p(C)\ge v(C)$ for all $C\subset I$, so $\vecp$ is stable.

To see that this bound is asymptotically
tight, consider a game $G'=(I, v)$ with $|I|=n=2k+1$
given by $v(C)=0$ if $|C|\le k$,
and $v(C)=1$ if $|C|\ge k+1$. Clearly, this game
is anonymous. Moreover, as any two winning coalitions intersect, 
this game is also super-additive. Consider any stable super-imputation $\vecp$
for this game. For any $C$ with $|C|=k+1$, we have $\sum_{i\in C}p_i\ge 1$.
There are exactly ${n\choose {k+1}}$ coalitions of this size, and each 
agent participates in exactly ${{n-1}\choose k}$ such coalitions. 
Thus, summing all these inequalities, we obtain 
${{n-1}\choose k}p(I)\ge {n\choose {k+1}}$, or, canceling, 
$p(I)\ge \frac{n}{k+1}=2-\frac{1}{k+1}$.~\end{proof}

A somewhat similar stability-related concept is the \emph{least core}, 
which is the set of all imputations $\vecp$ that minimize the maximal
\emph{deficit} $v(C)-p(C)$.
In particular, the \emph{value} of the least core $\eps(G)$,
defined as 
\[
\eps(G)=\inf_{\vecp\in\I(G)}\{\max \{v(C)-p(C)\mid C\subseteq I\}\},
\] 
is strictly positive if and only if the cost of stability 
is strictly positive. The following proposition 
provides a more precise description of the relationship
between the value of the least core and the cost of stability.

\begin{proposition}\label{cos-least-core}
For any coalitional game $G=(I, v)$ with $|I|=n$
such that $\eps(G)\ge 0$, we have 
$\CoS(G)\le n\eps(G)$, and this bound is asymptotically tight.
\end{proposition}

\begin{proof}
Clearly, if $\eps(G)=0$, we have $\CoS(G)=0$. Now, assume $\eps(G)>0$.
Let $\vecp$ be an imputation in the least core of $G$. For any
$C\subseteq I$, we have $p(C)\ge v(C)-\eps(G)$.
Consider a super-imputation $\vecp^*$ given by $p^*_i=p_i+\eps(G)$.
Clearly, we have $p^*(C)\ge v(C)$ for any $C\subseteq I$ such that
$C\neq\emptyset$,  
i.e., $\vecp^*$ is stable. Further, it is easy to see that
$p^*(I)=v(I)+n\eps(G)$, so $\CoS(G)\le n\eps(G)$. 

To see that this bound is asymptotically
tight, reconsider the game $G = (I,v)$ with $|I|
= n$, $v(\emptyset) = 0$, and $v(C) = 1$ for all $C \neq \emptyset$. It
is easy to see that $\eps(G) = \frac{n-1}{n}$, since the imputation
$(\frac{1}{n}, \ldots, \frac{1}{n})$ is in the least core of $G$. On the
other hand, as mentioned above, $\CoS(G) = n-1 = n\eps(G)$.~\end{proof}

\subsection{Algorithmic Properties of $\CoS(G)$}
\label{sec:compute}

The linear program $\lp$ 
provides a way of computing $\CoS(G)$ for any coalitional
game $G$. However, this linear program contains exponentially many constraints
(one for each subset of $I$). Thus, solving it directly would be too
time-consuming for most games. Note that for general 
coalitional games, this is, in a sense, inevitable: 
in general, a coalitional game is described
by its characteristic function, i.e., a list of $2^n$ numbers.
Thus, to discuss the algorithmic properties of $\CoS(G)$, 
we need to restrict our attention to games with compactly representable
characteristic functions. 

A standard approach to this issue is to consider 
games that can be described by polynomial-size circuits.
Formally, we say that a class $\calG$ of games has a 
\emph{compact circuit representation} if there exists
a polynomial $p$ such that for every $G\in \calG$,  
$G=(I, v)$, $|I|=n$, there exists a circuit $\calC$
of size $p(n)$ with $n$
binary inputs that on input $(b_1, \dots, b_n)$ outputs
$v(C)$, where $C=\{i\in I\mid b_i=1\}$.

Unfortunately, it turns out that having a compact
circuit representation does not guarantee efficient
computability of $\CoS(G)$. Indeed, it is easy to see that
WVGs with integer weights have such a representation.
However, in the next section we will show that computing
$\CoS(G)$ for such games is computationally intractable 
(Theorem~\ref{thm:cos-hard}). We can, however, provide a \emph{sufficient}
condition for $\CoS(G)$ to be efficiently computable. To do so, 
we will first formally state the relevant computational problems.

\smallskip

\noindent {\sc Super-Imputation-Stability}:\ 
Given a coalitional game $G$ (compactly represented by a circuit),
a supplemental payment
$\Delta$ and an imputation $\vecp=(p_1,\ldots,p_n)$ in the adjusted
game $G(\Delta)$, decide whether $\vecp\in\core(G(\Delta))$.

\smallskip

\noindent {\sc CoS}:\ 
Given a coalitional game $G$ (compactly represented by a circuit)
and a parameter $\Delta$,
decide whether $\CoS(G)\le \Delta$, i.e., whether
$\core(G(\Delta))\neq\emptyset$.

\smallskip

Consider first {\sc Super-Imputation-Stability}. 
Fix a game $G=(I, v)$. For any super-imputation $\vecp$ for $G$, 
let $d(G, \vecp)=\max_{C\subseteq I}(v(C)-p(C))$ be the maximum
deficit
of a coalition under $\vecp$. Clearly, $\vecp$ is stable
if and only if $d(G, \vecp)\le 0$. Observe also that for any $\Delta>0$
it is easy to decide whether $\vecp$ is an imputation for $G(\Delta)$.
Thus, a polynomial-time algorithm for computing $d(G, \vecp)$ can be converted
into a polynomial-time algorithm for {\sc Super-Imputation-Stability}.
Further, we can decide {\sc CoS} via solving $\lp$ by the ellipsoid method. 
The ellipsoid method runs in polynomial
time given a polynomial-time \emph{separation oracle}, i.e., a procedure
that takes as input a candidate feasible solution,
checks if it indeed is feasible, and if this is not the case,
returns a violated constraint. Now, given a vector $\vecp$
and a parameter $\Delta$,
we can easily check if they satisfy constraints (1)--(3),
i.e., if $\vecp$ is an imputation for $G(\Delta)$. To verify
constraint (4), we need to check if $\vecp$ is in the core
of $G(\Delta)$. As argued above, this can be done by checking
whether $d(G, \vecp)\le 0$. We summarize these results as follows.

\begin{theorem}
Consider a class of coalitional games $\calG$ with
a compact circuit representation.
If there is an algorithm that for any $G\in \calG$, $G=(I, v)$, 
$|I|=n$, and for any super-imputation $\vecp$ for $G$
computes $d(G, \vecp)$ in time $\poly(n, |\vecp|)$, 
where $|\vecp|$ is the number of bits in the binary representation of $\vecp$, 
then for any $G\in\calG$
the problems {\sc Super-Imputation-Stability} and {\sc CoS}
are polynomial-time solvable.
\end{theorem}

We mention in passing that for games 
with poly-time computable characteristic functions
both problems are in $\conp$. 
For {\sc Super-Imputation-Stability}, the membership is trivial;
for {\sc CoS}, it follows
from the fact that the game $G(\Delta)$ has a
poly-time computable characteristic function 
as long as $G$ does, 
and hence we can apply the results of~\cite{malizia07} 
(see the proof of 
Theorem~\ref{thm:cos-hard} for details).

\section{Cost of Stability in WVGs
Without Coalition Structures}

In this section, we focus on computing the cost of stabilizing the grand
coalition in WVGs. We start by considering the complexity of exact algorithms
for this problem. 
\subsection{Exact Algorithms}
\label{l_exact}

In what follows, unless specified otherwise, 
we assume that all weights and the threshold are integers given in binary, 
whereas all other numeric parameters, such as the supplemental payment
$\Delta$ and the entries of the payoff vector $\vecp$, are rationals given in 
binary. Standard results on linear threshold
functions~\cite{muroga} imply that WVGs with integer weights have a compact
circuit representation. Thus, we can define the computational problems
{\sc Super-Imputation-Stability-WVG} and {\sc CoS-WVG} 
by specializing the problems
{\sc Super-Imputation-Stability} and {\sc CoS} to WVGs.
Both of the resulting problems turn out to be computationally hard.

\begin{theorem}\label{thm:cos-hard}
The problems {\sc Super-Imputation-Sta\-bi\-li\-ty-WVG}
and {\sc CoS-WVG} are $\conp$-com\-plete.
\end{theorem}

\begin{proof}
Both of our reductions will be
from {\sc Partition}, a well-known $\np$-complete
problem~\cite{garey:np}, which is defined as follows:
given a list $A = (a_1,\ldots,a_n)$ of
nonnegative integers such that $\sum_{i=1}^na_i=2K$,
decide whether there is a sublist $A'$  of $A$ such that
$\sum_{a_i \in A'} a_i = K$.

We first show that {\sc CoS-WVG} is coNP-hard.
Given an instance $A = (a_1,\ldots,a_n)$ of {\sc Partition}, we construct
a weighted voting game $G$ by setting $I=\{1, \dots, n\}$, $w_i=a_i$
for each $i$, $1 \leq i \leq n$, and $q=K$. Set $\Delta=\frac{K-1}{K+1}$.
We claim that $(G, \Delta)$ is a ``yes''-instance of {\sc CoS-WVG}
if and only if $A$ is a ``no''-instance of {\sc Partition}.

Indeed, suppose that $A$ is a ``yes''-instance of {\sc Partition},
and let $A'$ be the corresponding sublist. Set $I'=\{i\mid a_i\in A'\}$
and $I''=I\setminus I'$. Suppose for the sake of contradiction
that $G(\Delta)$ has a nonempty core, and let $\vecp$ be an imputation
in the core of $G(\Delta)$. We have $p(I)=\frac{2K}{K+1}<2$, and hence
either $p(I')<1$ or $p(I'')<1$ (or both). On the other hand, since
$\sum_{i\in I'}a_i=K$, we have $w(I')=w(I'')=K=q$, i.e., at least
one of the coalitions $I'$ and $I''$ has a rational incentive to deviate,
a contradiction.

On the other hand, suppose that $A$ is a ``no''-instance of {\sc Partition},
and consider a vector $\vecp^*=(p^*_1, \dots, p^*_n)$, where
$p^*_i=\frac{w_i}{K+1}$. We have $p^*(I)=\frac{2K}{K+1}$, and hence
$p^*(I)-v(I)=\frac{K-1}{K+1}$. That is, $\vecp^*$ is an imputation
for $G(\Delta)$. We now show that $\vecp^*$ is in the core
of $G(\Delta)$ (and thus that $G(\Delta)$ has a nonempty core).
Indeed, consider any coalition $C\subset I$ such that $v(C)=1$.
We have $w(C)\ge q$. Moreover, as $A$
is a ``no''-instance of {\sc Partition},
there is no coalition $C\subset I$ whose weight is exactly $q$,  
so we have $w(C)\ge q+1=K+1$. Thus we have
$p^*(C)= \frac{w(C)}{K+1}\ge 1$. Hence, the agents in $C$
have no rational incentive to deviate from $\vecp^*$
and therefore $\vecp^*\in\core(G(\Delta))$.

We can use the same construction
to show that {\sc Super-Imputation-Stability-WVG} is $\conp$-hard.
Indeed, consider $G$, $\Delta=\frac{K-1}{K+1}$, and $\vecp^*$ defined above.   
It follows from our proof
that $\vecp^*$ is in the core of $G(\Delta)$ if and only if
$A$ is a ``no''-instance of {\sc Partition}.
Moreover, {\sc Super-Imputation-Stability-WVG} is clearly in $\conp$:
to verify that a given super-imputation $\vecp$ is unstable, it
suffices to guess a coalition $C$ and verify that it is winning,
i.e., $w(C)\ge q$, but is paid less than one under~$\vecp$.
Finally, to see that {\sc CoS-WVG} is in $\conp$, observe that
this problem is equivalent to deciding whether the corresponding
game $G(\Delta)$
has a nonempty core. Furthermore, it is easy to see
that $G(\Delta)$ has a {\em polynomial-time
compact representation} in the sense of Definition~3.1 
in~\cite{malizia07}. Thus, Theorem~5.3 in~\cite{malizia07}
implies that deciding whether the core of $G(\Delta)$ 
is nonempty is in $\conp$. Hence, {\sc CoS-WVG} is 
also in $\conp$.~\end{proof}

The reductions in the proof of Theorem~\ref{thm:cos-hard}
are from {\sc Partition}. Consequently, 
our hardness results depend in an essential way on the weights
being given in binary. Thus, it is natural to ask what happens 
if the agents' weights are polynomially
bounded (or given in unary). 
It turns out that in this case the results of Section~\ref{sec:compute}
imply that {\sc Super-Imputation-Stability-WVG} 
and {\sc CoS-WVG} are in~$\p$, since
for WVGs with small weights one can compute $d(G, \vecp)$
in polynomial time.

\begin{theorem}\label{thm:cos-easy} 
{\sc Super-Imputation-Stability-WVG} 
and {\sc CoS-WVG} are in $\p$ when
the agents' weights are polynomially bounded
(or given in unary).
\end{theorem}

\begin{proof}
As argued in Section~\ref{sec:compute}, it suffices to show
that given a WVG $G=[\w; q]$  
and a super-imputation $\vecp$ for $G$, 
we can compute $d(G, \vecp)$ in time $\poly(n, w_{\max}, |\vecp|)$, 
where $|\vecp|$ denotes the number of bits in the binary representation of 
$\vecp$.

For any $i$, $1 \leq i \leq n$, and any $w$, $1 \leq w \leq w(I)$,
let
$$
X_{i, w}=\min\{p(C)\mid C\subseteq \{1, \dots, i\}, w(C)=w\}.
$$
We can compute the quantities $X_{i, w}$ inductively as follows.
For $i=1$, we have $X_{i, w}=p_1$ if $w=w_1$, and $X_{i, w}=+\infty$
otherwise. Now, suppose that we have computed $X_{i', w}$
for each $i'$, $1 \leq i' \leq i$. We can then compute $X_{i+1, w}$
as $X_{i+1, w}=\min\{X_{i, w}, p_i+X_{i, w-w_i}\}$.
Observe that $p^*=\min\{X_{n, w}\mid w\ge q\}$
is the minimal payment that a winning coalition in $G$
can receive under $\vecp$.
As $p_i\ge 0$ for all $i$, $1 \leq i \leq n$, 
we have $d(G, \vecp)=1-p^*$.

Clearly, the running time of this algorithm
is polynomial in $n$, $w_{\max}$ and $|\vecp|$.
Observe that one can construct a similar algorithm that runs 
in polynomial time even if the weights are large,
as long as all entries of $\vecp$ can take polynomially
many values.~\end{proof}

\subsection{Approximating the Cost of Stability in Weighted Voting Games}
\label{l_approx}

For large weights, the algorithms outlined at the end of the previous section
may not be practical. Thus, the center may want to trade off its payment
and computation time, i.e., provide a slightly higher supplemental
payment for which the corresponding stable super-imputation can be computed 
efficiently. It turns out that this is indeed possible, i.e., 
$\CoS(G)$ can be efficiently approximated to an arbitrary
degree of precision.

\begin{theorem}\label{thm:fptas}
There exists an algorithm $\calA(G, \eps)$ that, 
given a WVG $G=[\w; q]$ 
in which the weights of all players
are nonnegative integers given in binary
and a parameter $\eps>0$, outputs a value 
$\Delta$ that satisfies $\CoS(G)\le \Delta\le (1+\eps)\CoS(G)$
and runs in time $\poly(n, \log w_{\max}, 1/\eps)$.
That is, there exists a fully polynomial-time approximation scheme
(FPTAS) for $\CoS(G)$.
\end{theorem}

\begin{proof}
We start by proving a simple lemma that will be useful
for the analysis of our algorithm.

\begin{lemma}\label{cosge1/n}
For any WVG $G$ such that
$\CoS(G)\neq 0$, we have $\CoS(G)\ge 1/n$.
\end{lemma}
\sproofof{Lemma~\ref{cosge1/n}}
Consider a weighted voting game $G$ that does not have a veto player
and hence $\CoS(G)\neq 0$. Suppose for the sake of contradiction that
$\CoS(G)=\Delta<1/n$, that is, the game $G(\Delta)$ has a nonempty
core. Let $\vecp=(p_1, \dots, p_n)$ be an imputation
in the core of $G(\Delta)$. As we have $v'(I)=\Delta+1>1$, there
must be at least one player $i$ such that $p_i>1/n$.
Hence, $p(I\setminus\{i\})<1+\Delta-1/n<1$. Therefore
the coalition $I\setminus\{i\}$ satisfies
$v(I\setminus\{i\})=1$ (since $i$ is not a veto player),
$p(I\setminus\{i\})<1$, and hence $\vecp$ is not stable,
a contradiction.~\eproofof{Lemma~\ref{cosge1/n}}

Our proof of the theorem
is inspired by the FPTAS
for the value of the least core of WVGs~\cite{conf/aaai/ElkindGGW07}.

We will first describe an additive fully polynomial-time
approximation scheme for $\CoS(G)$, i.e., 
an algorithm $\calA'(G, \eps)$ 
that, given a WVG $G=[w_1, \dots, w_n; q]$
and $\eps>0$, can compute a value $\Delta$ satisfying
$\CoS(G)\le \Delta\le \CoS(G)+\eps$ and runs in time
$\poly(n, \log w_{\max}, 1/\eps)$. We will then show how to
convert it into an FPTAS using Lemma~\ref{cosge1/n}.

Set $X=2\lceil 1/\eps\rceil$, and let $\eps'=1/X$.
We have $\eps/4\le \eps'\le \eps/2$.

Consider the linear program $\lp$ given in
Section~\ref{l_def}.
Instead of solving $\lp$ directly, 
we consider a family of linear feasibility programs (LFP)
$(\calL_i)_{i=1, \dots, nX}$, where the $k$th LFP
$\calL_k$ is given by
\begin{eqnarray*}
p_i&\ge& 0 \text{\quad for  } i=1, \dots, n,\\
p_1+\dots+p_n&\le& 1+\eps'k,\\
\sum_{i\in C}p_i&\ge& 1 \text{\quad for all } C\subseteq N \text{ such that }
\sum_{i\in C}w_i \ge q.
\end{eqnarray*}
As $\eps'nX= n$, it follows that at least
one of these LFPs has a feasible solution.
Now, let $k^*$ be the smallest value of $k$ for which
$\calL_k$ has a feasible solution. We have 
$\eps'(k^*-1)< \CoS(G)\le \eps'k^*$, or, equivalently, 
$\CoS(G)\le \eps'k^*\le \CoS(G)+\eps'$. Hence, by computing
$k^*$ we can obtain an additive $\eps'$-approximation
to $\CoS(G)$. Now, while it is not clear
if we could find $k^*$ in polynomial time, 
we will now show how to find a value $k$
that is guaranteed to be in the set $\{k^*, k^*+1\}$.

It is natural to approach this problem by trying
to successively solve $\calL_1, \dots, \calL_{nX}$.
However, just as the linear program $\lp$, 
the LFP $\calL_k$ has exponentially many
constraints (one for each winning coalition of $G$).
Moreover, an implementation of the separation oracle 
for $\calL_k$ would involve solving {\sc Knapsack}, 
which is an NP-hard problem when weights are given in binary.
Hence, we will now take a somewhat different approach.
Namely, we will show how to design an algorithm $\calS$ that, 
given a candidate solution $(p_1, \dots, p_n)$ for $\calL_k$,
either outputs a constraint that is violated by this solution
or finds a feasible solution for $\calL_{k+1}$.
The running time of $\calS(p_1, \dots, p_n)$
is $\poly(n, \log w_{\max}, 1/\eps)$.

The algorithm $\calS$ first checks if the candidate solution
$(p_1, \dots, p_n)$ satisfies the first $n+1$ constraints of the 
LFP. If no violated constraint is discovered at this step, 
it rounds up the payoffs by setting
$p'_i=\min\{\frac{\eps't}{n}\mid t\in{\mathbb N}, \frac{\eps't}{n}\ge p_i\}$ 
for each $i$, $1 \leq i \leq n$.
Note that for each $i$, $1 \leq i \leq n$, 
we have $p_i\le p'_i\le p_i+\frac{\eps'}{n}$, 
and the rounded payoff $p'_i$
can be represented as $p'_i=\frac{\eps'}{n}t_i$, where 
$t_i\in \{0, \dots, nX\}$.
We can now use a variant of the dynamic programming algorithm
used in the proof of Theorem~\ref{thm:cos-easy}
to decide whether
there is a subset of agents $C$ that satisfies 
$\sum_{i\in C}w_i\ge q$ and $\sum_{i\in C}p'_i < 1$
(see the remark at the end of that proof).
If there is such a subset, the rounded vector 
$(p'_1, \dots, p'_n)$ violates the constraint that corresponds to 
$C$, and hence the original vector $(p_1, \dots, p_n)$, 
which satisfies $p_i\le p'_i$ for all $i\in I$, 
violates it, too. Hence, $\calS$ outputs the 
corresponding constraint
and stops. Otherwise, it follows that $(p'_1, \dots, p'_n)$
satisfies all constraints of $\calL_k$
that correspond to the winning
coalitions of $G$. Moreover, we have
$$
\sum_{i=1}^n p'_i\le \sum_{i=1}^n p_i+n\frac{\eps'}{n}\le 1+\eps'k+\eps'.
$$
Hence, $(p'_1, \dots, p'_n)$ is a feasible solution for 
$\calL_{k+1}$, so $\calS$ outputs it and stops.

We are now ready to describe
our algorithm $\calA'$. It tries to solve
$\calL_1, \calL_2, \dots$ (in this order).
To solve $\calL_k$, 
it runs the ellipsoid algorithm on its input. Whenever the ellipsoid
algorithm makes a call to the separation oracle, $\calA'$ passes
this request to $\calS$, which either identifies a violated constraint, 
in which case $\calA'$ continues simulating the ellipsoid algorithm, 
or outputs a feasible solution for $\calL_{k+1}$, in which case
$\calA'$ stops and outputs $\eps'(k+1)$. If the ellipsoid algorithm
terminates and decides that the current LFP does not have a feasible 
solution, $\calA'$ proceeds to the next LFP in its list.
If the ellipsoid algorithm outputs a feasible solution for $\calL_k$, 
$\calA$ outputs $\eps'k$.

Recall that we denote by $k^*$ the smallest value of $k$ for which $\calL_{k}$
has a feasible solution.
Clearly, $\calA$ will correctly report
that neither of $\calL_1, \dots, \calL_{k^*-2}$ has a  feasible solution.
When solving $\calL_{k^*-1}$, it will either solve it correctly
(i.e., report that it has no feasible solutions) and move on to $\calL_{k^*}$, 
or discover a feasible solution for $\calL_{k^*}$. In the former case, 
$\calA'$ will either solve $\calL_{k^*}$ correctly, i.e., find
a feasible solution, or discover a feasible solution to $\calL_{k^*+1}$.
In either case, the output $\eps'k$ of our algorithm satisfies 
$k\in\{k^*, k^*+1\}$.

We have shown that 
$\CoS(G)\le \eps'k^*\le \CoS(G)+\eps'$.
Consequently, we have 
$\CoS(G)\le \eps'k\le \eps'(k^*+1)\le \CoS(G)+2\eps'\le \CoS(G)+\eps$.
This proves that $\calA'$
is an additive fully polynomial-time
approximation scheme for the cost of stability. 

We will now show how to convert $\calA'$ into an FPTAS $\calA$.
Our algorithm $\calA$ is given a game $G=[\w; q]$ and a parameter $\eps$.
It first tests if $\CoS(G)=0$ (equivalently, 
if $G$ has a nonempty core). By Theorem~\ref{veto}, 
this can be done by checking if $G$ has a veto player, i.e., 
whether $w(I\setminus\{i\})< q$ for some $i$, $1 \leq i \leq n$. 

If $\CoS(G)\neq 0$, $\calA$ runs $\calA'$
on input $(G, \eps/n)$. Let $\Delta$ be the output
of $\calA'(G, \eps/n)$; we have $\CoS(G)\le \Delta\le \CoS(G)+\eps/n$.
On the other hand, by Lemma~\ref{cosge1/n}
we have $\CoS(G)\ge 1/n$, and therefore
$$
\CoS(G)+\eps/n\le \CoS(G)+\eps \CoS(G)
=(1+\eps)\CoS(G).
$$
Hence $\Delta$ satisfies
$\CoS(G)\le \Delta\le (1+\eps)\CoS(G)$, as required.~\end{proof}

Moreover, one can get a 2-approximation to the 
adjusted gains simply by paying each agent in proportion to her weight,
and this bound can be shown to be tight.

\begin{theorem}\label{thm:2approx}
For any WVG $G=[\w; q]$ 
with $\CoS(G)=\Delta$, 
the super-imputation $\vecp^*$ given by
$p^*_i = \min \{1, \frac{w_i}{q}\}$ is stable
and satisfies $p^*(I)\le 2 p(I)$
for any super-imputation $\vecp\in\core(G(\Delta))$.
\end{theorem}

\begin{proof}
First, it is easy to see that $\vecp^*$ is stable, as
we have $p^*(C)\ge \min\{1, \frac{w(C)}{q}\}$. 

Now, set $\Delta=\CoS(G)$ and
fix a super-imputation $\vecp$ in the core of $G(\Delta)$.
Let $I'=\{i\mid w_i\ge q\}$ and set $k=|I'|$. Clearly, if
$i\in I'$, for any stable super-imputation $\vecp'$ we have $p'_i\ge 1=p^*_i$.
On the other hand, it is clear that paying any agent more than $1$
is suboptimal, so $p_i=1$ for any $i\in I'$.

Sort all agents in $I\setminus I'$ by decreasing weights, and partition them
into sets $C_1,\ldots,C_m$ in the following way:

\begin{itemize}
\item
Set $j = 0$.
\item
While there are unallocated agents:
\begin{itemize}
\item
Set $j = j + 1$;
\item
Add agents to $C_j$ until $w(C_j) \geq q$
 or until there are no more agents.
\end{itemize}
\item 
Set $m = j$.
\item 
If $w(C_j) \geq q$, set $m=j+1$ and $C_m=\emptyset$.
\end{itemize}

Note that this procedure guarantees that $w(C_m) < q$, i.e., 
the last coalition $C_m$ loses. In particular, 
if $m=1$ then $w(C_1)<q$. Since $w(I) \geq q$, this means that
$k \geq 1$ and $C_1=I\setminus I'$. In this case, we have 
$$
p(I) \ge k, \quad 
p^*(I) = k + \sum_{i \in C_1} \frac{w_i}{q} < k + \frac{q}{q} = k + 1, 
$$ 
and hence $p^*(I)/p(I) < (k+1)/k \leq 2$.
Therefore, throughout the rest of the proof we can assume $m > 1$. 

Set $j' = \arg\max_{j \leq m} w(C_j)$, that is, 
$j'$ is the index of a maximum-weight coalition among $C_1, \dots, C_m$.
Observe that since $w(C_1)\ge q$ and $w(C_m)<q$, we have $j'\neq m$. 
To finish the proof, we consider two cases and show that
$p^*(I) \leq 2p(I)$ holds in each of them.

\begin{description}
\item[Case~1:] $w(C_{j'})+w(C_m) \leq 2q$.
For each $j \leq m-1$, we have $w(C_j) \geq q$, and therefore $p(C_j) \geq 1$.
Thus, we have
$$
p(I) \ge k + \sum_{j \neq m} p(C_j) = k+m-1.
$$
On the other hand, we have $w(C_j) \leq 2q$ for all $j$, $1 \leq j \leq m$, so
\begin{eqnarray*}
p^*(I) 
&= &    p^*(I')+\sum_{j \neq j',m} p^*(C_j) + p^*(C_{j'})+p^*(C_m) \\
&\leq & k  +\sum_{j \neq j',m} \frac{w(C_j)}{q}+\frac{w(C_{j'})+w(C_m)}{q} \\
& \leq & k + 2(m-2) +2  \leq 2(k+m-1) \leq 2 p(I).
\end{eqnarray*}

\item[Case~2:]
$w(C_{j'})+w(C_m) > 2q$.
We begin by computing $p^*(I)$, as it may be slightly larger in this case:
\begin{eqnarray*}
p^*(I) & =   & k + \sum_{j\neq m}\frac{w(C_j)}{q} + \frac{w(C_m)}{q} \\
       & \le & k + \frac{(m-1)2q + q}{q} = k+2m-1. 
\end{eqnarray*}
Fortunately, we can provide a better lower bound for $p(I)$. 
Let $A_1$ be the set that contains the last player in $C_{j'}$ only,
and set $A_2 = C_{j'} \setminus A_1$ and $A_3 = C_m$.
We have $w(A_1)< q$, since $A_1$ has just one agent,
and we have already removed all agents whose weight is at least $q$.
Furthermore, we have $w(A_2) < q$, since we move on to the next set
as soon as a total weight of at least $q$ is reached in the current set.
On the other hand, we have $A=A_1 \cup A_2 \cup A_3=C_{j'}\cup C_m$. As
$w(C_{j'})+w(C_m) > 2q$, 
we have $w(A_1)+w(A_3) = w(A) - w(A_2) \geq 2q - q = q$
and 
 $w(A_2)+w(A_3) = w(A) - w(A_1) \geq 2q - q = q$.

Therefore, we have $p(A_1\cup A_2)\ge 1$, $p(A_1\cup A_3)\ge 1$, 
$p(A_2\cup A_3)\ge 1$, and hence $p(A_1 \cup A_2\cup A_3)\ge 3/2$.
Thus, we have
\begin{eqnarray*}
p(I)&=&\sum_{i \in I'}p_i + \sum_{j\neq j',m}p(C_j) + p(C_{j'}) +p(C_m) \\
    & \geq & k + (m-2) + p(C_{j'}\cup C_m) \\
    & = & k + m-2 + p(A_1\cup A_2 \cup A_3) \\
    &\geq &  k + m -2 + \frac32 = \frac12(2k+2m-1) \geq \frac12p^*(I).
\end{eqnarray*}
\end{description}
This completes the proof of Theorem~\ref{thm:2approx}.~\end{proof}

To see that the analysis presented above is tight, 
consider the game
$[1-\frac{\epsilon}{3},1-\frac{\epsilon}{3};1]$ for any fixed $\eps>0$.
We have  $p^*(I) = 2-\frac{2\eps}{3}$. On the other hand, 
this game has a nonempty core, so we have $p(I)=1$, and hence
$p^*(I) > (2-\epsilon)p(I)$.

\section{Cost of Stability in Games with Coalition Structures}
\label{l_cs}

If a coalitional game is not super-additive, the formation of the grand coalition 
is not necessarily the most desirable outcome: for example, it may be the case
that by splitting into several teams the agents can accomplish more tasks 
than by working together. In such settings, the central authority may want
to stabilize a coalition structure, i.e., a partition of agents 
into teams. We now generalize
the cost of stability to such settings.

\subsection{Stabilizing a Fixed Coalition Structure}

We first consider the setting where the central authority
wants to stabilize a particular coalition structure.

Given a coalitional game $G=(I, v)$, a coalition structure
$\CS=(C^1, \dots, C^m)$ over $I$
and a vector $\vecDelta=(\Delta^1, \dots, \Delta^m)$, let
$G(\vecDelta)$ be the game with the set of agents $I$
and the characteristic function $v'$ given by
$v'(C^i)=v(C^i)+\Delta^i$ for $i=1, \dots, m$ and $v'(C)=v(C)$
for any $C\not\in\{ C^1, \dots, C^m\}$.
We say that the game $G(\vecDelta)$ is \emph{stable
with respect to $\CS$} if there exists an imputation 
$\vecp\in\I(\CS, G(\vecDelta))$
such that $(\CS, \vecp)$
is in the CS-core of $G(\vecDelta)$.
Also, we say that an external payment $\Delta$ \emph{stabilizes}
a coalition structure $\CS$ with respect to a game $G$ if there exist
$\Delta^1\ge 0, \dots, \Delta^m\ge 0$ such that 
$\Delta=\Delta^1+\cdots+\Delta^m$ and the game
$G(\vecDelta)$ is stable with respect to $\CS$.
We are now ready to define the cost of stability
of a coalition structure $\CS$ in $G$.

\begin{definition}
Given a coalitional game $G=(I, v)$ and a coalition structure 
$\CS=(C^1, \dots, C^m)$ over $I$, 
the \emph{cost of stability} $\CoS(\CS, G)$ of the coalition structure 
$\CS$ in $G$
is the smallest external payment needed to stabilize $\CS$, i.e., 
\begin{align*}
\CoS(\CS, G) = \inf 
\{ 
  \sum_{i=1}^m\Delta^i \mid 
  & \Delta^i\ge 0 \text{ }\mathrm{for}\text{ }i=1, \dots, m
  \quad\mathrm{and}\\
  & \exists \vecp\in\I(\CS, G(\vecDelta))
  \quad\mathrm{ s.t. }\quad
  (\CS, \vecp)\in\cscore(G(\vecDelta))
 \}.
 \end{align*}
\end{definition}

Fix a game $G=(I, v)$ and set $v_{\max}=\max_{C\subseteq I}v(C)$.
It is easy to see that for any coalition structure 
$\CS=(C^1, \dots, C^m)$ 
the game $G(\vecDelta)$, where $\Delta^i=|C^i|v_{\max}$, 
is stable with respect to $\CS$, and therefore $\CoS(\CS, G)$
is well-defined and satisfies $\CoS(\CS, G)\le nv_{\max}$.
Moreover, as in the case of games without coalition structures, the value
$\CoS(\CS, G)$ can be obtained as an optimal solution to a linear program.
Indeed, we can simply take the linear program $\lp$ 
and replace the 
constraint $\sum_{i\in I}p_i = v(I)+\Delta$ with the constraint
$\sum_{i\in I}p_i = v(\CS)+\Delta$. It is not hard to see that the resulting
linear program, which we will denote by $\lp_\CS$, 
computes $\CoS(\CS, G)$: in particular, the constraints
$\Delta^i\ge 0$ for $i=1, \dots, m$ are implicitly captured by
the constraints 
$\sum_{i\in C^i}p_i\ge v(C^i)$ in line~(4) of $\lp_\CS$.

We now turn to the question of computing the cost of stability
of a given coalition structure in WVGs.
To this end, we will modify the decision problems 
stated in Section~\ref{l_exact} as follows.

\medskip

\noindent {\sc Super-Imputation-Stability-WVG-CS:}
Given a WVG $G=[\w;q]$ with the set of agents $I$,
a coalition structure $\CS=(C^1, \dots, C^m)$ over $I$,  
a vector $\vecDelta=(\Delta^1, \dots, \Delta^m)$ and an imputation
$\vecp\in\I(\CS, G(\vecDelta))$, 
decide if $(\CS, \vecp)$ is in the CS-core of
$G(\vecDelta)$.

\medskip

\noindent {\sc CoS-WVG-CS:}
Given a WVG $G=[\w;q]$ with the set of agents $I$, 
a coalition structure $\CS$ over $I$ and a parameter $\Delta$,
decide whether $\CoS(\CS, G)\le \Delta$.

\medskip

The results of Section~\ref{l_exact} immediately imply
that both of these problems are computationally hard even for $m=1$.
Moreover, using the results of~\cite{ecai08}, we can show that
{\sc Super-Imputation-Sta\-bi\-li\-ty-WVG-CS} remains $\conp$-complete
even if $\vecDelta$ is fixed to be $(0, \dots, 0)$.
On the other hand, when weights are integers given in unary, 
both {\sc CoS-WVG-CS} and {\sc Super-Imputation-Sta\-bi\-li\-ty-WVG-CS}
are polynomial-time solvable.
Indeed, to solve 
{\sc Super-Imputation-Stability-WVG-CS}, one needs to check
if there is a coalition $C$ with $w(C)\ge q$, $p(C)<1$.
This can be done using the dynamic programming algorithm
from the proof of Theorem~\ref{thm:cos-easy}.
Moreover, to solve {\sc CoS-WVG-CS}, we can simply
run the ellipsoid algorithm on the linear program $\lp_\CS$
described earlier in this section, using the algorithm for
{\sc Super-Imputation-Stability-WVG-CS} as a separation oracle.
Thus, we obtain the following result.
\begin{theorem}
When all players' weights are integers given in unary,
the problems  {\sc CoS-WVG-CS} and {\sc Super-Im\-pu\-ta\-tion-Sta\-bi\-li\-ty-WVG-CS}
are in $\p$.
\end{theorem}
Finally,
we 
adapt the approximation
algorithm presented in Section~\ref{l_approx} to this setting.
\begin{theorem}
There exists an FPTAS for $\CoS(\CS, G)$ in WVGs.
\end{theorem}
\subsection{Finding the Cheapest Coalition Structure to Stabilize}
So far, we have focused on the setting where the external party wants to
stabilize a particular coalition structure. However, it can also be the case
that the central authority simply wants to achieve stability, 
and does not care which coalition structure arises, as long as it can
be made stable using as little money as possible.
We will now introduce the notion of \emph{cost of stability for
games with coalition structures} to capture this type of setting.
Recall that $\calCS(I)$ denotes
the set of all coalition structures over $I$.
\begin{definition}
Given a coalitional game $G=(I, v)$, let 
the \emph{cost of stability for $G$ with coalition structures}, denoted
by $\CoS_\CS(G)$,
be $\min\{\CoS(\CS, G)\mid \CS\in\calCS(I)\}$.
\end{definition}
Clearly, one can compute $\CoS_\CS(G)$ by enumerating
all coalition structures over $I$ and picking the one with the
smallest value of $\CoS(\CS, G)$. Alternatively, note that the linear
program $\lp_\CS$ depends only on the value of the coalition
structure $\CS$. Hence, stabilizing  
all coalition structures with the same total value has the same cost.
Moreover, this implies that the cheapest coalition structure
to stabilize is the one that maximizes social welfare.
Hence, if we could compute the value of the coalition structure $\CS^*$ 
that maximizes social welfare, we could find $\CoS_\CS(G)$
by solving $\lp_{\CS^*}$.

For WVGs, paper~\cite{ecai08} (see Theorem~2 there) shows that
if weights are given in binary, it is NP-hard to decide whether 
a given game has a nonempty CS-core. As this question
is equivalent to asking whether $\CoS_\CS(G)=0$, the latter
problem is NP-hard, too. 
One might hope that computing $\CoS_\CS(G)$ is easy
if the weights of all players are given in unary. However, this does
not seem to be the case. Indeed, our algorithms for computing
the cost of stability in other settings relied on solving
the corresponding linear program. To implement
this approach in our scenario, we would need to compute
the value of the coalition structure that maximizes social welfare.
However, a straightforward reduction from {\sc 3-Partition}, 
a classic problem that is known to be NP-hard even for unary weights,  
shows that the latter problem is NP-hard even if weights are given in unary.
While this does not immediately imply that computing
$\CoS_\CS(G)$ is hard for small weights, it means that finding
the cheapest-to-stabilize outcome is NP-hard even if weights
are given in unary.

\subsection{Stabilizing a Particular Coalition}
\label{sec:single}

We now consider the case where the central
authority wants a particular group of agents
to work together, but does not care about the stability of the overall
game. Thus, it wants to identify a coalition structure
containing a particular coalition $C$ and the minimal subsidy
to the players that ensures that no set of players that includes members
of $C$ wants to deviate. We omit the formal definition 
of the corresponding cost of stability concept, as well as 
its algorithmic analysis due to space constraints.
However, we would like to mention several subtle
points that arise in this context.
First, one might think that the optimal way to stabilize a coalition
is to offer payments to members of this coalition only. However, 
this turns out not to be true, as the following example shows.

\begin{example}\label{ex1:single}
Consider the game $G=[1, 1, 1; 2]$ and the coalition $C=\{1, 2\}$.
If we were to stabilize $C$ by paying its members
only, we would have to ensure that each of them receives a payment
of 1, resulting in an external payment of $1$: if, e.g.,
player 1 receives $p_1<1$,
player $3$ could offer him
to form the coalition $\{1, 3\}$ and distribute the payoffs as
$p'_1=p_1+\frac{1-p_1}{2}>p_1$, $p'_3=\frac{1-p_1}{2}>0=p_3$.
On the other hand,
it is not hard to see that the payoff vector
$(\frac12, \frac12, \frac12)$ ensures that no group
of players wants to deviate from $(\{1, 2\}, \{3\})$,
i.e., the central authority can stabilize $C$ by spending
$\frac12$ only as long as it is willing to 
pay the players outside of $C$.
Thus, the cheapest way to stabilize a particular coalition 
may involve paying agents who do not belong to that coalition.
\end{example}

Second, as shown by Example~\ref{ex2:single} below,
stabilizing a given coalition
may be strictly cheaper than stabilizing \emph{any} of the 
coalition structures
that contain it. 
Thus choosing a good definition of the cost of stability
of an individual coalition is a nontrivial issue.

\begin{example}\label{ex2:single}
Consider the weighted voting game
$G=[8, 8, 9, 9, 1; 10]$ and a coalition $C=\{1, 2\}$.
It is not hard to check that $G$ has an empty CS-core
and therefore $\CoS_\CS(G)>0$.
However, no player in $C$ has an incentive to deviate
from the coalition structure $\CS=(\{1, 2\}, \{3, 4\}, \{5\})$
with the payoff vector $\vecp=(.5, .5, .5, .5, 0)$. That is, if the central
authority is only interested in stabilizing $C$, it can achieve this
goal without spending any money. However, from a long-term perspective
this approach may be dangerous. Indeed, consider the coalition $\{4, 5\}$
that has an incentive to deviate from $(\CS, \vecp)$.  
If this deviation happens, player $3$ is left on her own, and will
be happy to form a coalition with player $1$ in which, e.g., $1$ gets $.9$
and $3$ gets $.1$. Clearly, this proposition would be attractive to player
$1$ as well, which would cause the coalition $C$ to fall apart.
Thus, stabilizing a given coalition
may be strictly cheaper than stabilizing \emph{any} of the
coalition structures
that contain it.
\end{example}

\section{Related Work}
\label{l_related_work}

The complexity of various solution concepts
in coalitional games
is a well-studied topic~\cite{cs04,is05,cs06,yokoo05}.
In particular, \cite{conf/aaai/ElkindGGW07} analyzes some important
computational aspects of stability 
in WVGs, proving a number of results on the complexity of the least 
core and the nucleolus.  The complexity of the CS-core in WVGs is studied 
in~\cite{ecai08}. Paper~\cite{journals/jair/MondererT04} is similar to 
ours in spirit. 
It considers the setting 
where an external party intervenes in order to achieve a certain
outcome using monetary payments.
However, \cite{journals/jair/MondererT04} 
deals with the very different domain of \emph{non}cooperative games.
There are also similarities between our work and the recent
research on bribery in
elections~\cite{fal06}, 
where an external party pays voters to change their preferences
in order to make a given candidate win.
A companion paper~\cite{mfcs-cos} studies
the cost of stability in network
flow games.

\section{Conclusion}
\label{l_conclusions}

We have examined the possibility of stabilizing a coalitional
game by offering the agents additional payments in order 
to discourage them from deviating, and defined the cost of stability as the minimal
total payment that allows a stable division of the gains. 
We focused on the computational aspects of this concept
for weighted voting games. In the setting where the outcome 
to be stabilized is the grand coalition, we provided a complete
picture of the computational complexity of the related decision problems.
We then extended our results to settings where
agents can form a coalition structure. 

There are several lines of possible future research. First, 
while the focus of this paper was on weighted voting games, 
the notion of the cost of stability is defined for any coalitional
game. Therefore, a natural research direction
is to study the cost of stability
in other classes of games. Second, we would like to develop
a better understanding of the relationship between the cost of stability
of a game, and its least core and nucleolus.
Finally, it would be interesting to extend the notion
of the cost of stability to games with nontransferable utility
and partition function games.

\medskip
\noindent{\bf Acknowledgments:} We thank the SAGT'09 reviewers for many
helpful comments.  This work was supported in part 
by 
DFG grants RO~\mbox{1202/\{11-1, 12-1\}}, the ESF's EUROCORES program LogICCC,
the Alexander von Humboldt Foundation's TransCoop program, 
EPSRC grant GR/T10664/01, ESRC grant  ES/F035845/1, 
NRF Research Fellowship, 
Singapore Ministry of Education AcRF Tier~1 Grant, 
and by ISF grant \#898/05.
This work was done in part while the first author was at Hebrew University
and while the sixth author was visiting Hebrew University and the
University of Rochester.

\bibliographystyle{plain}

\end{document}